\documentclass[11pt]{article}
\usepackage[letterpaper]{geometry}
\usepackage{amsmath,amsthm,amsfonts,amssymb}
\usepackage{enumerate,color,xcolor}
\usepackage{graphicx}
\usepackage{subfigure}
\usepackage{url}
\usepackage{comment}
\numberwithin{equation}{section}
\theoremstyle{plain}

\newtheorem{theorem}{Theorem}

\newtheorem{lemma}[theorem]{Lemma}

\newtheorem{proposition}[theorem]{Proposition}

\newtheorem{remark}[theorem]{Remark}

\begin{document}

\begin{center}
  \Large \bf Option Pricing for the Variance Gamma Model: A New Perspective
\end{center}

\author{}
\begin{center}
Yuanda Chen\,\footnote{Department of Mathematics, Florida State University, 1017 Academic Way, Tallahassee, FL-32306, United States of America; ychen@math.fsu.edu},
Zailei Cheng\,\footnote{Department of Mathematics, Florida State University, 1017 Academic Way, Tallahassee, FL-32306, United States of America; zc12b@my.fsu.edu},
Haixu Wang\,\footnote{Department of Mathematics, Florida State University, 1017 Academic Way, Tallahassee, FL-32306, United States of America; hwang@math.fsu.edu}
\end{center}

\begin{center}
 \today
\end{center}

\begin{abstract}
The variance gamma model is a widely popular model for option pricing in both academia and industry. 
In this paper, we provide a new perspective for pricing European style options
for the variance gamma model by deriving closed-form formulas combining
the randomization method and fractional derivatives. We also compare our results
with various existing results in the literature by numerical examples.
\end{abstract}

\section{Introduction}

The variance gamma process is a pure jump L\'{e}vy process that is widely used in option pricing. The
process is obtained by evaluating arithmetic Brownian motion with drift at a random time given by a gamma process with unit mean rate and certain variance rate\footnote{See Section~\ref{sec:vg} for definitions.}.
The variance gamma process was first introduced by Madan and Seneta \cite{MS} as a candidate
for modeling the underlying stock market returns, as an alternative to the classical Black-Scholes
model. This alternative was sought primarily to obtain a process consistent with the observation
that the local movement of log stock prices is heavy-tailed relative to the normal distribution,
while the movement over large time intervals approaches normalities.
The variance gamma process was first considered in the context of option pricing in Madan
and Milne (1991) \cite{MM}, where it was used to price European options.
By comparing with the Black-Scholes model, they found that the variance gamma option values
are higher, in particular, for out-of-the-money options with large maturity on stocks
with high means, low variances and high kurtosis.
As the variance gamma process is obtained by evaluating Brownian motion with drift
at a random time governed by a gamma process, the additional parameters, i.e. the drift
of the Brownian motion, and the volatility of the time change, provide control over
the skewness and kurtosis of the return distribution. Madan et al. \cite{Madan} estimated
the statistical and risk-neutral densities using the S\&P 500 data, and they observed that
the statistical density is symmetric with some kurtosis, while the risk neutral density is negatively skewed with a larger kurtosis. The additional parameters also correct for pricing biases of the Black-Scholes model.
 
To obtain prices of European options under a variance gamma model, 
Madan and Milne (1991) \cite{MM} derived the option price as a weighted average
of Black-Scholes formula with the maturity averaged with respect to a gamma density.
Madan et al. (1998) \cite{Madan} built on the formula in \cite{MM}
and after applying a series of change of variables, they derived an alternative
analytical formula for the option price using some special functions.
Carr and Madan (1998) \cite{CarrMadan} proposed to use the fast Fourier transform method, 
and as a special case, the variance gamma model was studied.
Fu (2000) \cite{Fu} surveyed Monte Carlo methods for the variance gamma model, 
and discussed three methods for sequential simulation of the variance gamma process, two bridge sampling methods,
variance reduction via importance sampling, and estimation of the Greeks.
Cont and Voltchkova (2005) \cite{Cont} described a finite difference scheme for solving the partial
integro-differential equation arising from pricing European options under a variance gamma model.
For pricing American options and path-dependent options under variance gamma models, we refer the readers to \cite{Hirsa, AO, Kaishev2009} and references therein for further details. 

In this paper, we provide a new alternative analytical formula for pricing European style options
for the variance gamma model, see Theorem~\ref{MainThm}. The closed-form formula is derived by 
combining the Carr's randomization method \cite{Carr98} and fractional derivatives. 
The variance gamma model can be viewed as the Black-Scholes model with a random maturity
which follows a gamma distribution. Carr's randomization method was originally applied to compute the (American) option price
for Black-Scholes model by approximating a finite maturity by a sequence of Erlang distributed
random maturities. Carr \cite{Carr98} first derived the formula for the exponential distributed maturity,
and by noting that the Erlang density is a derivative of integer orders of the exponential density, 
the formulas for the Erlang distributed random maturities were then derived. Since the variance gamma model can be viewed
as the Black-Scholes model with maturity randomized by a gamma distribution,
Carr's randomization method provides explicit formula for a subclass of the variance
gamma model, that is when the maturity is Erlang distributed\footnote{The Erlang distribution 
is a subclass of the gamma distribution.}. We apply Carr's randomization method 
to price European options for the variance gamma model, and the formula is exact when 
the random maturity is Erlang distributed. To fill in the gap between the Erlang maturity and the more
general gamma maturity, we use the tools from the fractional calculus.
This methodology is intuitive since the Erlang density can be obtained by taking derivatives of integer
orders from the exponential density, and for the more general gamma density, fractional derivatives
are the natural tools. We will introduce a non-conventional fractional derivative that is tailored to the
particular application in our context to do the work.

Fractional derivatives and fractional calculus have been used in the
literature on option pricing. For example, Cartea and del-Castillo-Negrete
\cite{Cartea2007} showed that European-style option prices under particular
L\'{e}vy processes (including the CGMY model) are governed by a fractional
partial differential equation (FPDE) with two spatial fractional
derivatives capturing the non-locality induced by pure jumps in the
underlying price. See also \cite{BoLe2002, Cartea2013, ChenXuZhu2014,
Fallahgoul2016, ItkinCarr2012, MaromMomoniat2009} and the references
therein for applications of fractional calculus in option pricing. In these
studies, the fractional derivative is typically taken with respect to
either space (log stock price) or time (time to maturity). Our work differs
from these studies in that there is no FPDE involved and the fractional
derivative is taken with respect to the Laplace exponent of the Laplace transform of the option price.

We also compare our results
with various existing methods in the literature by numerical examples.
More precisely, we compare our results with the option pricing formula in Madan et al. \cite{Madan} (MCC)
and the fast Fourier transform method in Carr and Madan \cite{CarrMadan} (FFT).
The runtime of our method is comparable with both MCC and FFT methods. 
When the maturity is Erlang distributed, our method is consistently
faster across different regimes.
Finally, in the short-maturity regimes, our method beats both MCC and FFT methods.

\subsection{Organization of the paper}
The rest of the paper is organized as follows. Section~2 reviews the definitions of variance gamma processes and fractional derivatives. Section~3 presents our main results on alternative closed-form formulas for European options under a variance gamma model. Section~4 discusses numerical experiments based on our new formulas and compare its efficiency with exsiting methods. Auxiliary technical proofs are collected in the appendix.

\section{Preliminaries}
\subsection{A review of variance gamma processes} \label{sec:vg}
This section reviews the definition of variance gamma processes. See e.g. \cite{Madan} for details. 

Let us consider a Brownian motion with drift model $X({t})$ which has an infinitesimal generator given by:
\begin{equation} \label{eq:generator}
\mathcal{L}f(x)=\mu f'(x)+\frac{1}{2}\sigma^{2}f''(x),
\end{equation}
for any $f$ in the domain of the infinitesimal generator.
The variance gamma process is obtained by evaluating the process $X$ at a random time change given by a gamma process
$\{\gamma(t): t \ge 0 \}$ with mean rate one, variance rate $\nu$, and $\gamma(0)=0.$ Note the gamma process is a L\'{e}vy process with independent gamma increments over non-lapping time intervals. The density of the increment on the time interval $(t, t+h)$ is given by the gamma density function with shape parameter $\alpha = h/\nu$ and rate parameter $\beta = 1/\nu$, i.e., the density is
\begin{equation}
f_h(x)= \left(\frac{1}{\nu}\right)^{\frac{h}{\nu}} \frac{x^{h/\nu - 1} \exp(-\frac{1}{\nu} x)}{\Gamma(h/\nu)}, \quad \text{for $x>0$.}
\end{equation}
The variance gamma process is then defined by
\begin{eqnarray}
X (\gamma(t)), \quad \text{for $t \ge 0$.}
\end{eqnarray}
where $\gamma(t)$ is a Gamma random variable with shape parameter $\alpha = t/\nu$ and rate parameter $\beta = 1/\nu.$ Hence $\mathbb{E}[\gamma(t)] = t$ and $Var(\gamma(t))= t \nu.$ 
In the variance gamma model, the stock price $S(t)$ is defined
as the exponential of the variance gamma process, that is,
\begin{equation}
S(t)=e^{X(\gamma(t))}.
\end{equation}


\subsection{Fractional Derivatives}
In this section, we introduce the definition of fractional derivatives and their properties that we will work with in the paper.

Let us define the $\alpha$-th order fractional derivative as\footnote{In the literature, there are many ways
to define fractional derivatives. For example, a more common
definition is given by $\hat{D}^{\alpha}_{x}f(x)=\frac{1}{\Gamma(1-\alpha)}\frac{d}{dx}
\int_{-\infty}^{x}(x-t)^{-\alpha}f(t)dt$, which is known 
as the Riemann-Liouville Fractional Derivative, see e.g. \cite{MR}. 
This definition can be applied to compute
the fractional derivative of $e^{\lambda x}$, where $\lambda>0$,
so that $\hat{D}^{\alpha}_{x}e^{\lambda x}=\lambda^{\alpha}e^{\lambda x}$, 
but it can not be applied to $e^{-\lambda x}$ that we want, where $\lambda>0$,
in the sense that $\hat{D}^{\alpha}_{x}e^{-\lambda x}$ does not exist.}
\begin{equation}
D^{\alpha}_{x}f(x)=\frac{1}{\Gamma(1-\alpha)}\frac{d}{dx}
\int_{x}^{\infty}(t-x)^{-\alpha}f(t)dt,
\end{equation}
where $\alpha<1$ provided that $D^{\alpha}_{x}f(x)$ exists and is finite.
When $\alpha$ is a non-negative integer, it coincides with the ordinary derivatives of order $\alpha$. 
For $n<\alpha<n+1$, where $n$ is a non-negative integer, $D^{\alpha}_{x}$ can be equivantly operated by
first applying the ordinary derivatives $\frac{d^{n}}{dx^{n}}$ and then applying the fractional derivative $D^{\alpha-n}_{x}$.

We now present two results on properties of fractional deratives. Lemma~\ref{expLemma} will be useful in our theoretical analysis, while Lemma~\ref{ImplementLemma} will be useful for 
efficient numerical computations of option prices under the Variance Gamma model.

\begin{lemma}\label{expLemma}
For any $\lambda>0$ and $\alpha<1$,
$D^{\alpha}_{x}e^{-\lambda x}=-\lambda^{\alpha}e^{-\lambda x}$.
\end{lemma}

\begin{lemma}\label{ImplementLemma}
Let $\alpha<1$. Assume that $f$ is twice differentiable, 
and for every $x$, $\lim_{t\rightarrow\infty}f(t)(t-x)^{1-\alpha}
=\lim_{t\rightarrow\infty}f'(t)(t-x)^{1-\alpha}=0$.
Then, we have
\begin{equation}
D^{\alpha}_{x}f(x)
=\frac{-1}{\Gamma(1-\alpha)(1-\alpha)}x^{2-\alpha}\int_{0}^{1}y^{\alpha-3}\left(1-y\right)^{1-\alpha}f''\left(\frac{x}{y}\right)dy.
\end{equation}
\end{lemma}

The proofs of these two lemmas are given in the appendix.

\section{Option Pricing for the Variance Gamma Model}

In this section, we present a new method 
based on combining Carr's randomization and fractional derivatives for pricing vanila European options 
with the Variance Gamma process as the underlying proecess. 

Without loss of generality, we assume that there is zero interest rate and no dividend yield. We are interested in computing the European put option price:
\begin{equation}\label{eq:put}
P(S, K, t) =\mathbb{E}_{S(0)=S}[(K-S(t))^{+}],
\end{equation}
where $K$ is the strick price, $t$ is the maturity and $S(t)=e^{X(\gamma(t))}$ is the stock price process introduced in Section~\ref{sec:vg}. 
The European call option price can be computed via the put-call parity.

Recall that $X(\cdot)$ is a Brownian motion with drift $\mu$ and volatility $\sigma$, and its infinitesimal generator is given in \eqref{eq:generator}.
Define
\begin{equation}
u(t,x):=\mathbb{E}_{X_{0}=x}[(K-e^{X_{t}})^{+}].
\end{equation}
It is clear that $u(t,x)$ satisfies the equation:
\begin{equation}
\frac{\partial u}{\partial t}
=\mathcal{L}u,
\end{equation}
with the initial condition $u(0,x)=(K-e^{x})^{+}$.
Taking the Laplace transform with respect to the time variable $t$, we get
\begin{equation}
m(\lambda,x):=\int_{0}^{\infty}\mathbb{E}_{X_{0}=x}[(K-e^{X_{t}})^{+}]e^{-\lambda t}dt=\int_{0}^{\infty}u(t, x)e^{-\lambda t}dt
\end{equation}
satisfies the equation:
\begin{equation}\label{eq:pde}
\lambda m-(K-e^{x})^{+}=\mathcal{L}m.
\end{equation}

By using dominated convergence theorem, one can show that
\begin{equation}
\lim_{x\rightarrow-\infty}m(\lambda,x)=\frac{K}{\lambda},
\qquad
\lim_{x\rightarrow+\infty}m(\lambda,x)=0.
\end{equation}

In addition, for all $n \ge 1$, 
$\left|\frac{\partial^{n}}{\partial\lambda^{n}}u(t, x)e^{-\lambda t}\right|\leq K\lambda^{n}$, uniformly in $(t,x)$. 
Thus, $m(\lambda,x)$ is infinitely differentiable w.r.t. $\lambda$ and
we find that the $n$-th derivative of $m(\lambda,x)$ with respect to (w.r.t.) $\lambda$ is given by
\begin{equation} \label{eq:m-n}
m^{(n)}(\lambda,x):=\frac{\partial^{n}}{\partial\lambda^{n}}m(\lambda,x)= (-1)^n \int_{0}^{\infty}u(t, x)e^{-\lambda t} t^n dt.
\end{equation}

We first present a result which gives analytical expressions of $m^{(n)}(\lambda,x)$ for all $n \ge 0$. 
The main idea of the proof is to use the approach in Carr \cite{Carr98} and the PDE in \eqref{eq:pde} to
iteratively solve the equation:
\begin{equation}
\lambda m^{(n)}+nm^{(n-1)}=\mathcal{L}m^{(n)},
\end{equation}
given $m^{(n-1)}$, $n\geq 1$. The proof is lengthy so we leave the details to the appendix.

\begin{proposition}\label{m0Thm}
We have
\begin{equation}
m(\lambda,x)
=\begin{cases}
\frac{K}{\lambda(\theta_{1}-\theta_{2})}e^{\theta_{1}(x-\log K)}-\frac{1}{\lambda}e^{x}+\frac{K}{\lambda} &\text{for $x\leq\log K$},
\\
\frac{K}{\lambda(\theta_{1}-\theta_{2})}e^{\theta_{2}(x-\log K)} &\text{for $x>\log K$}.
\end{cases}
\end{equation}
In addition, for integer $n \ge 1,$ the $n$-th derivative of $m(\lambda,x)$ w.r.t. $\lambda$ is given by
\begin{align*}
&m^{(n)}(\lambda,x)
\\
&=\begin{cases}
c_{11}^{(n)}e^{\theta_{1}x}+c_{12}^{(n)}xe^{\theta_{1}x}
+\cdots c_{1(n+1)}^{(n)}x^{n}e^{\theta_{1}x}
-\frac{(-1)^{n}n!}{\lambda^{n+1}}e^{x}+\frac{(-1)^{n}n!K}{\lambda^{n+1}} &\text{for $x\leq\log K$},
\\
c_{21}^{(n)}e^{\theta_{2}x}+c_{22}^{(n)}xe^{\theta_{2}x}
+\cdots c_{2(n+1)}^{(n)}x^{n}e^{\theta_{2}x} &\text{for $x>\log K$},
\end{cases}
\end{align*}
where the coefficients can be determined recursively as
\begin{equation}
nc_{i1}^{(n-1)}=\mu c_{i2}^{(n)}
+\sigma^{2}c_{i3}^{(n)}+\sigma^{2}c_{i2}^{(n)}\theta_{i},
\end{equation}
and for $j=2,3,\ldots,n$,
\begin{align*}
nc_{ij}^{(n-1)}
&=-\lambda c_{ij}^{(n)}
+\mu c_{ij}^{(n)}\theta_{i}
+\mu c_{i(j+1)}^{(n)}j
\\
&\qquad
+\frac{1}{2}\sigma^{2} c_{ij}^{(n)}(\theta_{i})^{2}
+\frac{1}{2}\sigma^{2} c_{i(j+2)}^{(n)}(j+1)j
+\frac{1}{2}\sigma^{2} c_{i(j+1)}^{(n)}2j\theta_{i},
\end{align*}
with $c_{i(n+2)}^{(n)}:=0$, and finally  
\begin{align}\label{cn11}
&c_{11}^{(n)}=\frac{\theta_{2}\sum_{j=2}^{n+1}(\log K)^{j-1}(c_{2j}^{(n)}K^{\theta_{2}}-c_{1j}^{(n)}K^{\theta_{1}})
-\frac{(-1)^{n}n!}{\lambda^{n+1}}K}{(\theta_{2}-\theta_{1})K^{\theta_{1}}}
\\
&\qquad\qquad
-\frac{\sum_{j=2}^{n+1}(\log K)^{j-1}(c_{2j}^{(n)}\theta_{2}K^{\theta_{2}}-c_{1j}^{(n)}\theta_{1}K^{\theta_{1}})}
{(\theta_{2}-\theta_{1})K^{\theta_{1}}}
\nonumber
\\
&\qquad\qquad\qquad\qquad
-\frac{\sum_{j=2}^{n+1}(j-1)(\log K)^{j-2}(c_{2j}^{(n)}K^{\theta_{2}}-c_{1j}^{(n)}K^{\theta_{1}})}{(\theta_{2}-\theta_{1})K^{\theta_{1}}},
\nonumber
\end{align}
and
\begin{align}\label{cn21}
&c_{21}^{(n)}=\frac{\theta_{1}\sum_{j=2}^{n+1}(\log K)^{j-1}(c_{2j}^{(n)}K^{\theta_{2}}-c_{1j}^{(n)}K^{\theta_{1}})
-\frac{(-1)^{n}n!}{\lambda^{n+1}}K}{(\theta_{2}-\theta_{1})K^{\theta_{2}}}
\\
&\qquad\qquad
-\frac{\sum_{j=2}^{n+1}(\log K)^{j-1}(c_{2j}^{(n)}\theta_{2}K^{\theta_{2}}-c_{1j}^{(n)}\theta_{1}K^{\theta_{1}})}
{(\theta_{2}-\theta_{1})K^{\theta_{2}}}
\nonumber
\\
&\qquad\qquad\qquad\qquad
-\frac{\sum_{j=2}^{n+1}(j-1)(\log K)^{j-2}(c_{2j}^{(n)}K^{\theta_{2}}-c_{1j}^{(n)}K^{\theta_{1}})}{(\theta_{2}-\theta_{1})K^{\theta_{2}}},
\nonumber
\end{align}
where $\theta_{1}>0>\theta_{2}$ are given by:
\begin{equation}
\theta_{1}:=\frac{-\mu+\sqrt{\mu^{2}+2\lambda\sigma^{2}}}{\sigma^{2}},
\qquad
\theta_{2}:=\frac{-\mu-\sqrt{\mu^{2}+2\lambda\sigma^{2}}}{\sigma^{2}}.
\end{equation}
\end{proposition}


With this proposition, we are now able to use fractional derivatives to give an explicit expression of European put option prices under the Variance Gamma model. The main result of this paper is given as follows. 

\begin{theorem}\label{MainThm}
Let $n\leq t/\nu-1<n+1$ for some $n\in\mathbb{N}\cup\{0,-1\}$, then
\begin{equation}
P(S,K,t)=\left(\frac{1}{\nu}\right)^{t/\nu}\frac{(-1)^{n+1}}{\Gamma(t/\nu)}D^{t/\nu-1-n}_{\lambda}
m^{(n)}(\lambda,\log S)\bigg|_{\lambda=1/\nu},
\end{equation}
where $D^{t/\nu-1-n}_{\lambda}$ is defined in Lemma \ref{ImplementLemma}
and $m^{(n)}(\lambda,\log S)$ is given in Proposition~\ref{m0Thm}.
\end{theorem}

\begin{remark}
The Variance Gamma model can be viewed as the Black-Scholes model with random maturity
that follows a gamma distribution. When $t/\nu$ is an integer, the random maturity
has an Erlang distribution, and $D^{t/\nu-1-n}_{\lambda}
m^{(n)}(\lambda,\log S)=m^{(t/\nu-1)}(\lambda,\log S)$ in Theorem \ref{MainThm},
which has an exact expression given in Proposition \ref{m0Thm}.
When $t/\nu$ is a fraction, the fractional derivative $D^{t/\nu-1-n}_{\lambda}$ 
in Theorem \ref{MainThm} can be computed using Lemma \ref{ImplementLemma}
\end{remark}

\begin{proof}[Proof of Theorem \ref{MainThm}]
Recall that $S(t)=e^{X(\gamma(t))}$ and we want to compute:
\begin{equation}
P(S,K,t)=
\int_{0}^{\infty}\mathbb{E}_{X_{0}=\log S}[(K-e^{X_{s}})^{+}]\left(\frac{1}{\nu}\right)^{t/\nu}\frac{s^{t/\nu-1}e^{-\frac{s}{\nu}}}{\Gamma(t/\nu)}ds.
\end{equation}
By definition
\begin{equation}
m(\lambda,x)=\int_{0}^{\infty}\mathbb{E}_{X_{0}=x}[(K-e^{X_{s}})^{+}]e^{-\lambda s}ds.
\end{equation}

Let us assume that $-1<t/\nu-1<1$ otherwise
one can first differentiate $m(\lambda,x)$ with respect to $\lambda$ integer times
and then apply the fractional derivative.

It is easy to see that 
\begin{equation}
D^{\alpha}_{\lambda}m(\lambda,x)
=-\int_{0}^{\infty}\mathbb{E}_{X_{0}=x}[(K-e^{X_{s}})^{+}]s^{\alpha}e^{-\lambda s}ds.
\end{equation}


Hence, we have
\begin{align*}
P(S,K,t)&=
\int_{0}^{\infty}\mathbb{E}_{X_{0}=\log S}[(K-e^{X_{s}})^{+}]\left(\frac{1}{\nu}\right)^{t/\nu}\frac{s^{t/\nu-1}e^{-\frac{s}{\nu}}}{\Gamma(t/\nu)}ds
\\
&=\left(\frac{1}{\nu}\right)^{t/\nu}\frac{(-1)}{\Gamma(t/\nu)}D^{t/\nu-1}_{\lambda}m(\lambda,\log S)\bigg|_{\lambda=1/\nu},
\end{align*}
and the above formula works for $-1<t/\nu-1<1$.


More generally, if $n\leq t/\nu-1<n+1$ for some $n\in\mathbb{N}\cup\{0,-1\}$, then we have
\begin{align*}
P(S,K,t)&=
\int_{0}^{\infty}\mathbb{E}_{X_{0}=\log S}[(K-e^{X_{s}})^{+}]\left(\frac{1}{\nu}\right)^{t/\nu}\frac{s^{t/\nu-1}e^{-\frac{s}{\nu}}}{\Gamma(t/\nu)}ds
\\
&=\left(\frac{1}{\nu}\right)^{t/\nu}\frac{(-1)^{n+1}}{\Gamma(t/\nu)}D^{t/\nu-1-n}_{\lambda}
m^{(n)}(\lambda,\log S)\bigg|_{\lambda=1/\nu}.
\end{align*}
The proof is therefore complete.
\end{proof}


\section{Numerical Experiments}


In this section, we numerically implement the put option price formula
we obtained in Theorem \ref{MainThm}, and compare our results
with some existing methods in the literature, including the alternative
closed-form formula in Madan et al. \cite{Madan}, the Fast Fourier Transform method 
in Carr and Madan \cite{CarrMadan}.


We compare the numerical results (CGZ) with
the Carr and Madan \cite{CarrMadan} Fast Fourier Transform method (FFT)
and the closed-form formula in Madan et al. \cite{Madan} (MCC), 
and summarize the comparision results in Table \ref{Table1}, 
Table \ref{Table2}, Table \ref{Table3} and Table \ref{Table4},  
and Figure \ref{fig1} and Figure \ref{fig2}.

In Table \ref{Table1}, Table \ref{Table2} and Figure \ref{fig1}, we fix
the model parameters to be $K=20$, $S=18$, $\sigma=0.1$ and $\nu=0.2$.
This is the in-the-money case for the put option. The three methods
give precisely the same price for the put options, and their price difference
is negligible. In terms of runtime, the three methods are comparable. 
The FFT method beats MCC all the time. Our method (CGZ)
has a clear advantage when $t/\nu$ is an integer, and
when it is not an integer, our results still show very decent numerical performance.

In Table \ref{Table3}, Table \ref{Table4}, and Figure \ref{fig2}, we fix
the model parameters to be $K=20$, $S=22$, $\sigma=0.1$ and $\nu=0.2$.
This is the out-of-the-money case for the put option.
We obtain similar phenomena as in the in-the-money case 
in Table \ref{Table1}, Table \ref{Table2} and Figure \ref{fig1}.

\begin{table}[h]
	\begin{center}
		\renewcommand{\arraystretch}{1.2}
		\begin{tabular}{ccccc}
			\hline
			$t$&Put Price&\multicolumn{3}{c}{Runtime}\\
			\hline
			&&MCC&FFT&CGZ\\
			\hline\hline
			0.2&    2.0107&    0.0063&    0.0016&    0.0002\\
			0.4&    2.0339&    0.0061&    0.0011&    0.0002\\
			0.6&    2.0662&    0.0064&    0.0010&    0.0002\\
			0.8&    2.1038&    0.0061&    0.0009&    0.0003\\
			1.0&    2.1441&    0.0063&    0.0010&    0.0003\\
			
			\hline
		\end{tabular}
		\vspace{1em}
		\caption{Runtime comparison with $K=20$, $S=18$, $\sigma=0.1$ and $\nu=0.2$ (integer $t/\nu$).}
		\label{Table1}
	\end{center}
\end{table}

\begin{table}[h]
	\begin{center}
		\renewcommand{\arraystretch}{1.2}
		\begin{tabular}{ccccc}
			\hline
			$t$&Put Price&\multicolumn{3}{c}{Runtime}\\
			\hline
			&&MCC&FFT&CGZ\\
			\hline\hline
			0.1&    2.0037&    0.0070&    0.0023&    0.0117\\
			0.3&    2.0209&    0.0065&    0.0014&    0.0122\\
			0.5&    2.0492&    0.0074&    0.0010&    0.0196\\
			0.7&    2.0845&    0.0068&    0.0010&    0.0281\\
			0.9&    2.1237&    0.0064&    0.0009&    0.0355\\
			
			\hline
		\end{tabular}
		\vspace{1em}
		\caption{Runtime comparison with $K=20$, $S=18$, $\sigma=0.1$ and $\nu=0.2$ (fractional $t/\nu$).}
		\label{Table2}
	\end{center}
\end{table}

\begin{table}[h]
	\begin{center}
		\renewcommand{\arraystretch}{1.2}
		\begin{tabular}{ccccc}
			\hline
			$t$&Put Price&\multicolumn{3}{c}{Runtime}\\
			\hline
			&&MCC&FFT&CGZ\\
			\hline\hline
			0.2&    0.0163&    0.0064&    0.0015&    0.0002\\
			0.4&    0.0489&    0.0067&    0.0009&    0.0002\\
			0.6&    0.0919&    0.0065&    0.0010&    0.0002\\
			0.8&    0.1401&    0.0063&    0.0010&    0.0003\\
			1.0&    0.1903&    0.0062&    0.0009&    0.0003\\
			
			\hline
		\end{tabular}
		\vspace{1em}
		\caption{Runtime comparison with $K=20$, $S=22$, $\sigma=0.1$ and $\nu=0.2$ (integer $t/\nu$).}
	\label{Table3}	
	\end{center}
\end{table}

\begin{table}[h]
	\begin{center}
		\renewcommand{\arraystretch}{1.2}
		\begin{tabular}{ccccc}
			\hline
			$t$&Put Price&\multicolumn{3}{c}{Runtime}\\
			\hline
			&&MCC&FFT&CGZ\\
			\hline\hline
			0.1&    0.0058&    0.0066&    0.0022&    0.0116\\
			0.3&    0.0309&    0.0067&    0.0014&    0.0124\\
			0.5&    0.0695&    0.0068&    0.0010&    0.0195\\
			0.7&    0.1156&    0.0068&    0.0010&    0.0276\\
			0.9&    0.1650&    0.0064&    0.0010&    0.0361\\
			
			\hline
		\end{tabular}
		\vspace{1em}
		\caption{Runtime comparison with $K=20$, $S=22$, $\sigma=0.1$ and $\nu=0.2$ (fractional $t/\nu$).}
		\label{Table4}
	\end{center}
\end{table}


\begin{figure}
\centering
\begin{subfigure}
  \centering
  \includegraphics[width=2.8in,clip=true,trim=100 200 100 200]{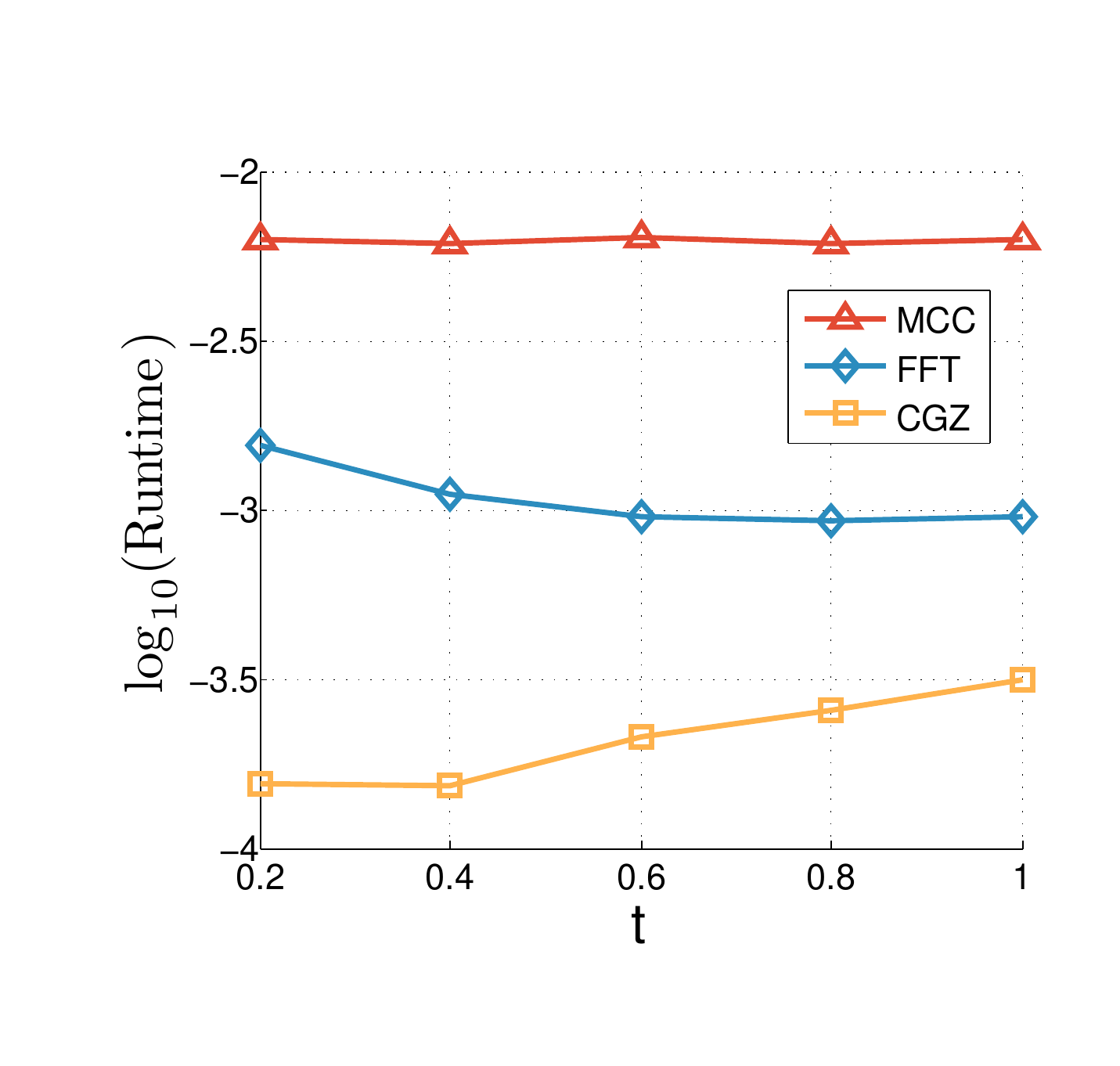}
\end{subfigure}
\begin{subfigure}
  \centering
  \includegraphics[width=2.8in, clip=true, trim=100 200 100 200]{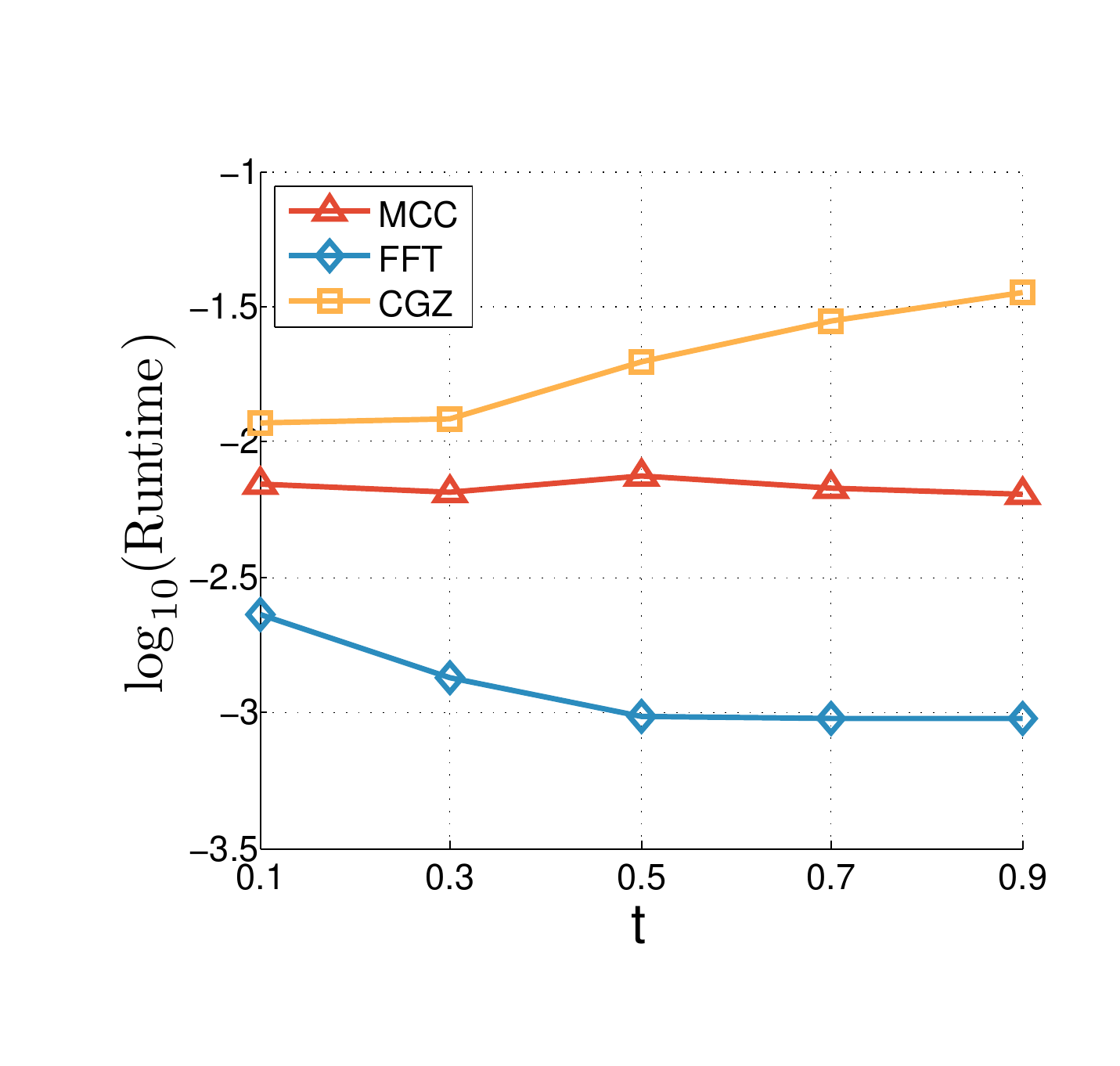}
\end{subfigure}
\caption{Runtime comparison with $K=20$, $S=18$, $\sigma=0.1$ and $\nu=0.2$. 
The picture on the left is when $t/\nu$ is an integer, and
the picture on the right is when $t/\nu$ is a fraction.}
\label{fig1}
\end{figure}

\begin{figure}
\centering
\begin{subfigure}
  \centering
  \includegraphics[width=2.8in,clip=true,trim=100 200 100 200]{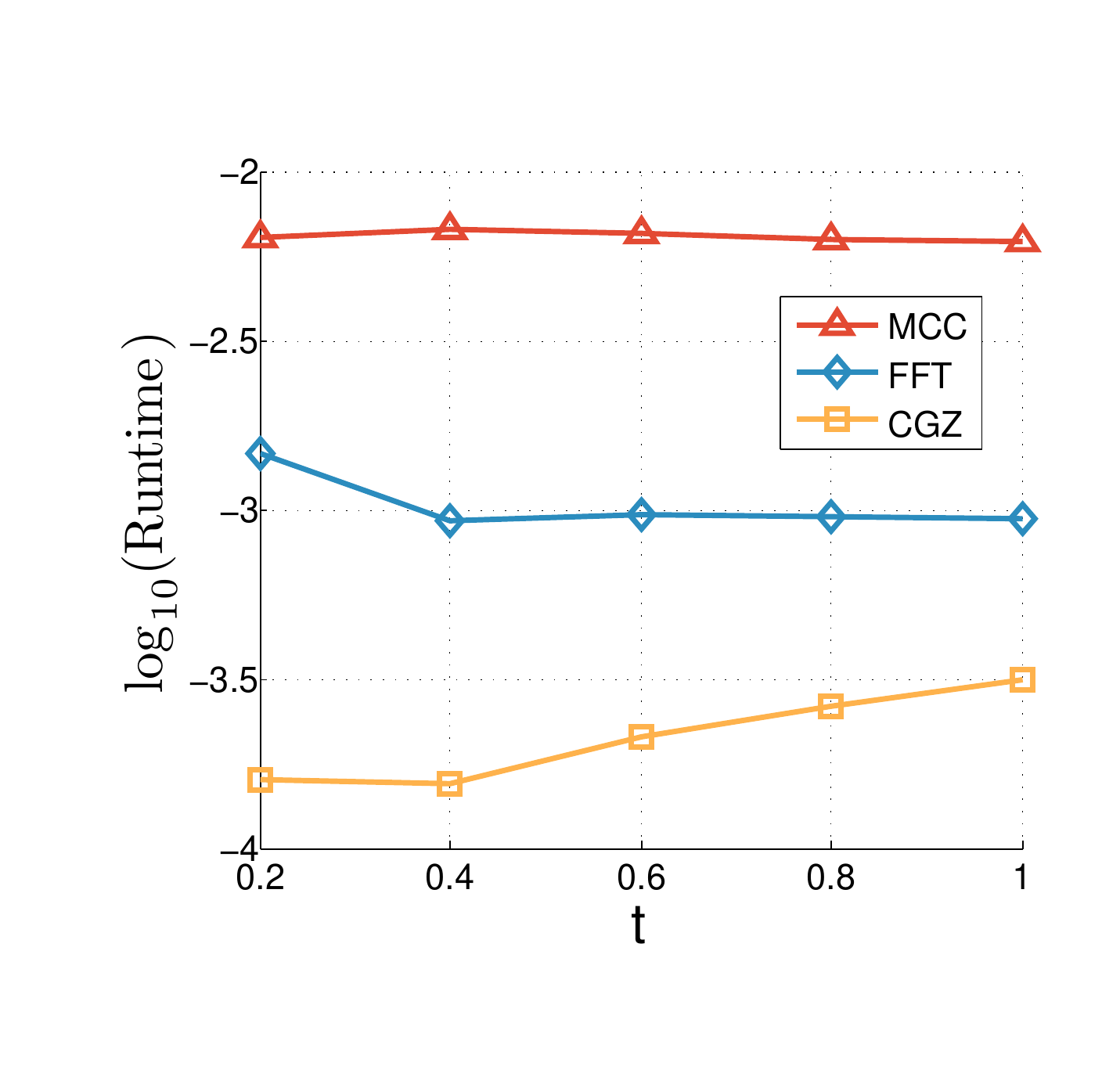}
\end{subfigure}
\begin{subfigure}
  \centering
  \includegraphics[width=2.8in,clip=true,trim=100 200 100 200]{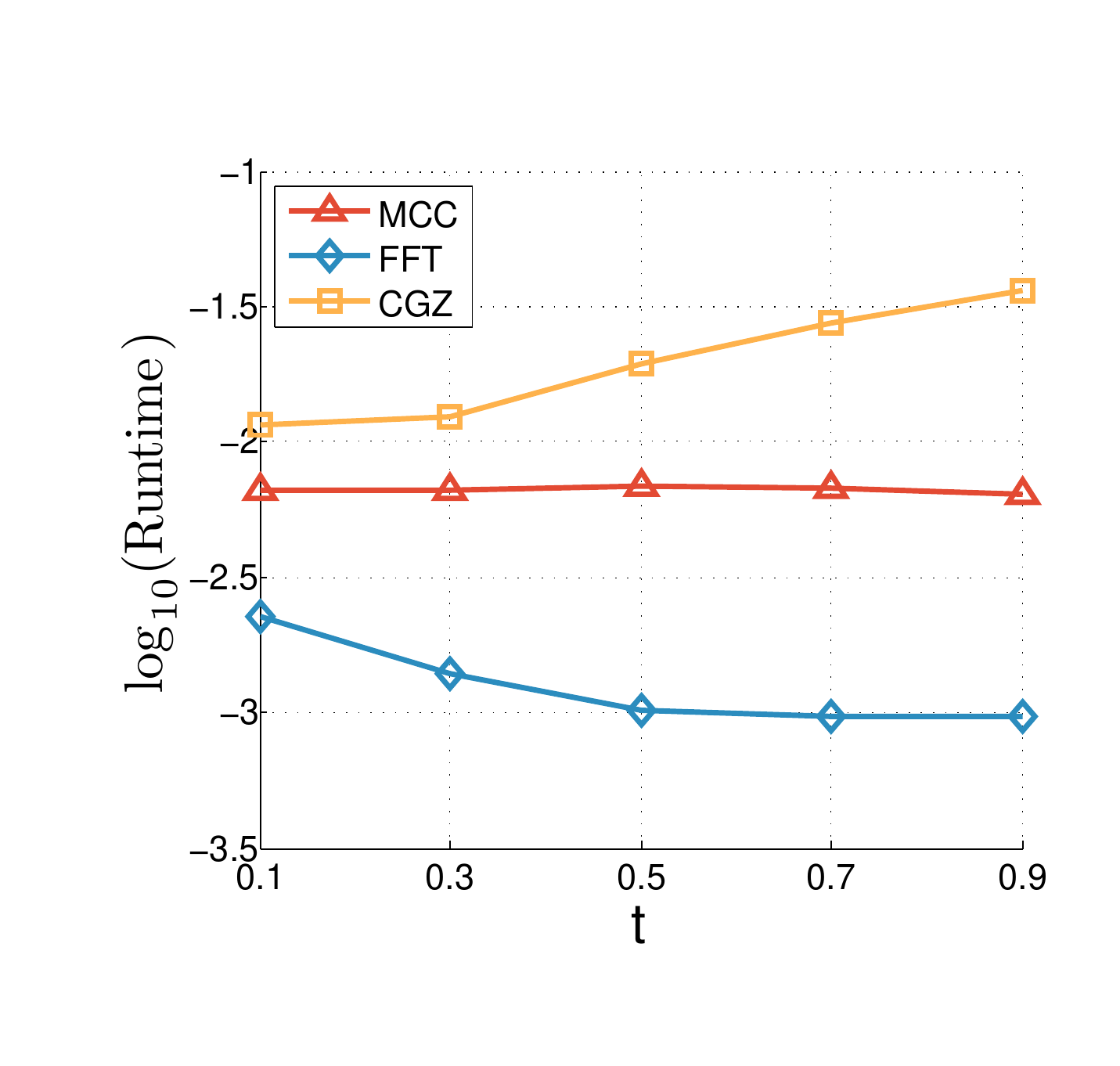}
\end{subfigure}
\caption{Runtime comparison with $K=20$, $S=22$, $\sigma=0.1$ and $\nu=0.2$.
The picture on the left is when $t/\nu$ is an integer, and
the picture on the right is when $t/\nu$ is a fraction.}
\label{fig2}
\end{figure}




\begin{table}[h]
	\begin{center}
		\renewcommand{\arraystretch}{1.2}
		\begin{tabular}{ccccc}
			\hline
			$t$&Put Price&\multicolumn{3}{c}{Runtime}\\
			\hline
			&&MCC&FFT&CGZ\\
			\hline\hline
			0.1000&    0.0020 &   0.0173 &   0.0039&    0.0115\\
			0.1200  &  0.0027 &   0.0126 &   0.0032 &   0.0117\\
			0.1400  &  0.0034 &   0.0117 &   0.0027 &   0.0116\\
			0.1600 &   0.0043 &   0.0111 &   0.0022 &   0.0117\\
			0.1800  &  0.0052 &   0.0098  &  0.0022 &   0.0114\\
			0.2000  &  0.0063 &   0.0092 &   0.0022 &   0.0114\\
			\hline
		\end{tabular}
		\vspace{1em}
		\caption{Runtime comparison for short maturity with $K=35$, $S=50$, $\sigma=0.2$ and $\nu=0.25$ (varying $t$). Price differences are within 1e-5.}
		
	\end{center}
\end{table}

\begin{table}[h]
	\begin{center}
		\renewcommand{\arraystretch}{1.2}
		\begin{tabular}{ccccc}
			\hline
			$t$&Put Price&\multicolumn{3}{c}{Runtime}\\
			\hline
			&&MCC&FFT&CGZ\\
			\hline\hline
			0.0500 &      0.0026&    0.0437&    0.0210 &   0.0138\\
			0.0700 &      0.0038&    0.0435&    0.0209 &   0.0122\\
			0.0900 &      0.0051&    0.0434&    0.0124  &  0.0126\\
			0.1100 &      0.0065&    0.0438&    0.0069   & 0.0122\\
			0.1300 &      0.0081&    0.0442&    0.0067    &0.0129\\
			0.1500 &      0.0097&    0.0439&    0.0065    &0.0127\\
			0.1700 &      0.0115&    0.0460&    0.0052    &0.0130\\
			0.1900 &   0.0134*&    0.0183&    0.0032    &0.0129\\
			\hline
		\end{tabular}
		\vspace{1em}
		\caption{Runtime comparison for short maturity with $K=35$, $S=50$, $\sigma=0.2$ and $\nu=0.5$ (varying $t$). MCC yields NaN for all rows except the last one (denoted by *) 
  Price differences are within 1e-5 (between FFT and CGZ) and 1.2e-5 between the non-NaN one from MCC and CGZ.}
		
	\end{center}
\end{table}

\begin{figure}
\centering
\begin{subfigure}
  \centering
  \includegraphics[width=2.8in,clip=true,trim=100 200 100 200]{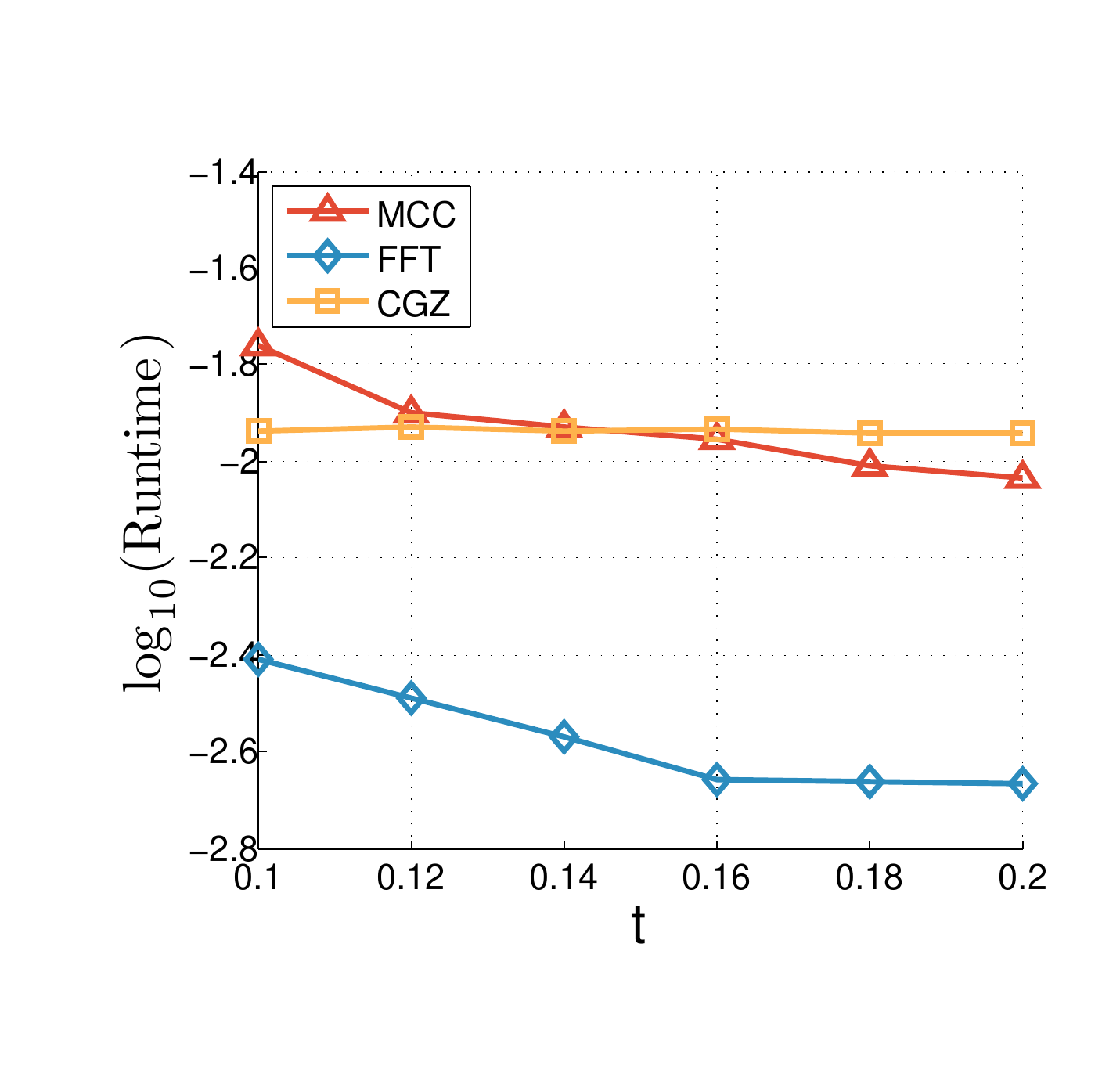}
\end{subfigure}
\begin{subfigure}
  \centering
  \includegraphics[width=2.8in,clip=true,trim=100 200 100 200]{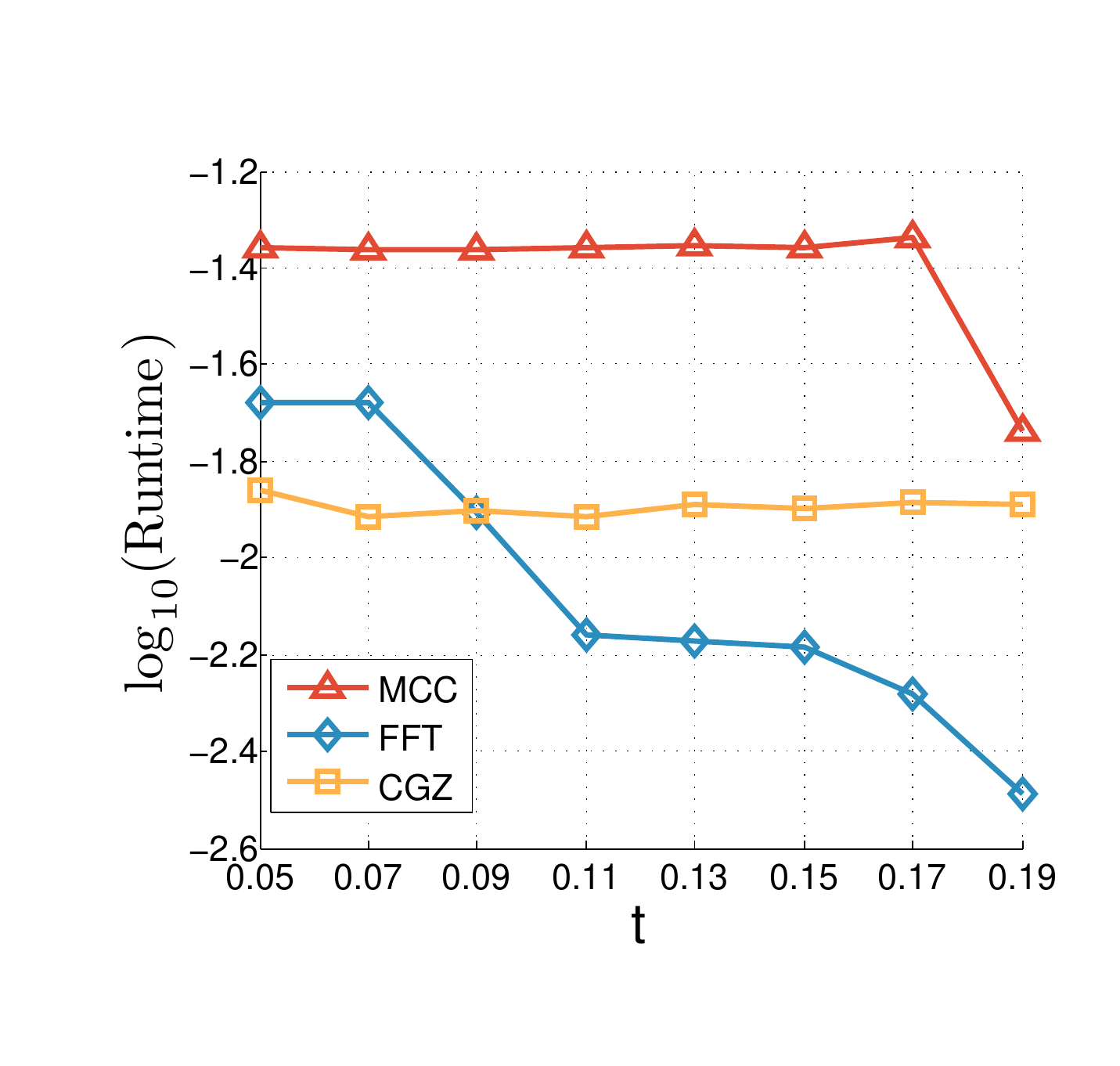}
\end{subfigure}
\caption{Runtime comparison for short maturity with $K=35$, $S=50$, $\sigma=0.2$. 
The picture on the left has $\nu=0.25$ varying $t$ and the picture
on the right has $\nu=0.5$ varying $t$.}
\label{fig3}
\end{figure}

\clearpage
\newpage
\section{Appendix: Proofs}

\begin{proof}[Proof of Lemma \ref{expLemma}]
Let $f(x)=e^{-\lambda x}$. Then by the definition of fractional derivatives we can compute that for $\alpha<1$,
\begin{align*}
D^{\alpha}_{x}f(x)
&=\frac{1}{\Gamma(1-\alpha)}\frac{d}{dx}\int_{x}^{\infty}(t-x)^{-\alpha}f(t)dt
\\
&=\frac{1}{\Gamma(1-\alpha)}\frac{d}{dx}\int_{x}^{\infty}(t-x)^{-\alpha}e^{-\lambda t}dt
\\
&=\frac{1}{\Gamma(1-\alpha)}\frac{d}{dx}x^{1-\alpha}\int_{1}^{\infty}(t-1)^{-\alpha}e^{-\lambda xt}dt
\\
&=\frac{1}{\Gamma(1-\alpha)}\frac{d}{dx}x^{1-\alpha}e^{-\lambda x}\int_{0}^{\infty}t^{-\alpha}e^{-\lambda xt}dt
\\
&=\frac{1}{\Gamma(1-\alpha)}\frac{d}{dx}x^{1-\alpha}e^{-\lambda x}(\lambda x)^{\alpha-1}
\int_{0}^{\infty}t^{-\alpha}e^{-t}dt
\\
&=\frac{d}{dx}e^{-\lambda x}\lambda^{\alpha-1}
\\
&=-\lambda^{\alpha}e^{-\lambda x}.
\end{align*}
Hence the proof is complete.
\end{proof}


\begin{proof}[Proof of Lemma \ref{ImplementLemma}]
From the definition of the fractional derivative, we have
\begin{align*}
D^{\alpha}_{x}f(x)
&=
\frac{1}{\Gamma(1-\alpha)}\frac{d}{dx}
\frac{1}{1-\alpha}\int_{x}^{\infty}f(t)d(t-x)^{1-\alpha}
\\
&=-\frac{1}{\Gamma(1-\alpha)}\frac{d}{dx}
\frac{1}{1-\alpha}\int_{x}^{\infty}(t-x)^{1-\alpha}f'(t)dt
\\
&=\frac{1}{\Gamma(1-\alpha)}\int_{x}^{\infty}(t-x)^{-\alpha}f'(t)dt
\\
&=\frac{-1}{\Gamma(1-\alpha)(1-\alpha)}
\int_{x}^{\infty}(t-x)^{1-\alpha}f''(t)dt.
\end{align*}
By changing variable $t=\frac{x}{y}$, we get $dt=\frac{-x}{y^{2}}dy$
and since $t$ is between $x$ and $\infty$, we have $y$ is between $0$ and $1$, and 
as a result, we get:
\begin{equation}
D^{\alpha}_{x}f(x)
=\frac{-1}{\Gamma(1-\alpha)(1-\alpha)}\int_{0}^{1}\left(\frac{x}{y}-x\right)^{1-\alpha}f''\left(\frac{x}{y}\right)\frac{x}{y^{2}}dy.
\end{equation}
The proof is therefore complete. 
\end{proof}


\begin{proof}[Proof of Proposition~\ref{m0Thm}]
We first compute $m(\lambda, x)$. 
For $x>\log K$, we have
\begin{equation}
\lambda m=\mu m'(x)+\frac{1}{2}\sigma^{2}m''(x),
\end{equation}
and for $x\leq\log K$, we have
\begin{equation}
\lambda m+e^{x}-K=\mu m'(x)+\frac{1}{2}\sigma^{2}m''(x).
\end{equation}
We can solve this equation and get:
\begin{equation}
m^{(0)}(\lambda,x)=m(\lambda,x)=c_{11}^{(0)}e^{\theta_{1}x}+c_{12}^{(0)}e^{\theta_{2}x}
+\left(\frac{K}{\lambda}+c_{13}^{(0)}e^{x}\right),
\end{equation}
for $x\leq\log K$, and
\begin{equation}
m^{(0)}(\lambda,x)=m(\lambda,x)=c_{21}^{(0)}e^{\theta_{2}x}+c_{22}^{(0)}e^{\theta_{1}x},
\end{equation}
for $x>\log K$, where $\theta_{1}>0>\theta_{2}$
are the solutions of the equation:
\begin{equation}
\mu\theta+\frac{1}{2}\sigma^{2}\theta^{2}-\lambda=0,
\end{equation}
and
\begin{equation}
c_{13}^{(0)}=\frac{-1}{\lambda-\mu-\frac{1}{2}\sigma^{2}},
\end{equation}

Note that
\begin{equation}
F(\theta):=\mu\theta+\frac{1}{2}\sigma^{2}\theta^{2}
-\lambda
\end{equation}
is convex in $\theta$ and $F(0)<0$ and $F(-\infty)=F(+\infty)=\infty$,
which implies that $F(\theta)=0$ has exactly two solutions,
one positive and the other one negative.
Hence we have two solutions $\theta_{1}>0>\theta_{2}$.

Next, let us show that $c_{12}^{(0)}=c_{22}^{(0)}=0$.
This comes from the boundary condition for $m(\lambda,x)$
as $x\rightarrow-\infty$ and $x\rightarrow+\infty$.

Hence, we conclude that
\begin{equation}
m(\lambda,x)
=\begin{cases}
c_{11}^{(0)}e^{\theta_{1}x}+c_{13}^{(0)}e^{x}+\frac{K}{\lambda} &\text{for $x\leq\log K$},
\\
c_{21}^{(0)}e^{\theta_{2}x} &\text{for $x>\log K$}.
\end{cases}
\end{equation}

We can then use $m(\lambda,x)$ being $C^{1}$ at $x=\log K$
to determine the coefficients $c_{11}^{(0)},c_{21}^{(0)}$:
\begin{align*}
&c_{11}^{(0)}K^{\theta_{1}}+c_{13}^{(0)}K+\frac{K}{\lambda}=c_{21}^{(0)}K^{\theta_{2}},
\\
&c_{11}^{(0)}\theta_{1}K^{\theta_{1}}+c_{13}^{(0)}K=c_{21}^{(0)}\theta_{2}K^{\theta_{2}}.
\end{align*}

Under the risk neutral measure with zero interest rate and zero dividend yield,
the $S(t)=e^{X(\gamma(t))}$ process is a martingale, and as a result, 
\begin{equation}
\mu+\frac{1}{2}\sigma^{2}=0,
\end{equation}
and thus $c_{13}^{(0)}=\frac{-1}{\lambda}$.
We can then compute that
\begin{equation}
c_{11}^{(0)}=\frac{K^{1-\theta_{1}}}{\lambda(\theta_{1}-\theta_{2})},
\qquad
c_{21}^{(0)}=\frac{K^{1-\theta_{2}}}{\lambda(\theta_{1}-\theta_{2})}.
\end{equation}

We next proceed to compute the first order derivative $m^{(1)}(\lambda,x)$ with respect to $\lambda$.
Recall that
\begin{equation}
\lambda m^{(1)}+m^{(0)}=\mu\frac{\partial}{\partial x}m^{(1)}
+\frac{1}{2}\sigma^{2}\frac{\partial^{2}}{\partial x^{2}}m^{(1)}.
\end{equation}
This implies that
\begin{equation}
m^{(1)}(\lambda,x)
=\begin{cases}
c_{11}^{(1)}e^{\theta_{1}x}+c_{12}^{(1)}xe^{\theta_{1}x}+\frac{1}{\lambda^{2}}e^{x}-\frac{K}{\lambda^{2}} &\text{for $x\leq\log K$},
\\
c_{21}^{(1)}e^{\theta_{2}x}+c_{22}^{(1)}xe^{\theta_{2}x} &\text{for $x>\log K$}.
\end{cases}
\end{equation}
For $x\leq\log K$,
we can compute that
\begin{align*}
&\mu\frac{\partial}{\partial x}c_{12}^{(1)}xe^{\theta_{1}x}
+\frac{1}{2}\sigma^{2}\frac{\partial^{2}}{\partial x^{2}}c_{12}^{(1)}xe^{\theta_{1}x}
-\lambda c_{12}^{(1)}xe^{\theta_{1}x}
\\
&=\mu c_{12}^{(1)}e^{\theta_{1}x}
+\sigma^{2}c_{12}^{(1)}\theta_{1}e^{\theta_{1}x}
=c_{11}^{(0)}e^{\theta_{1}x},
\end{align*}
if we choose
\begin{equation}
c_{12}^{(1)}=\frac{c_{11}^{(0)}}{\mu+\sigma^{2}\theta_{1}}.
\end{equation}
Similarly, we should choose
\begin{equation}
c_{22}^{(1)}=\frac{c_{21}^{(0)}}{\mu+\sigma^{2}\theta_{2}}.
\end{equation}
Finally, $c_{11}^{(1)}$ and $c_{21}^{(1)}$ are chosen
so that $m^{(1)}(\lambda,x)$ is $C^{1}$ in $x$ at $x=\log K$,
which yields that
\begin{align*}
&c_{11}^{(1)}K^{\theta_{1}}+c_{12}^{(1)}(\log K)K^{\theta_{1}}
=c_{21}^{(1)}K^{\theta_{2}}+c_{22}^{(1)}(\log K)K^{\theta_{2}},
\\
&c_{11}^{(1)}\theta_{1}K^{\theta_{1}}+c_{12}^{(1)}K^{\theta_{1}}
+c_{12}^{(1)}(\log K)\theta_{1}K^{\theta_{1}}+\frac{K}{\lambda^{2}}
=c_{21}^{(1)}\theta_{2}K^{\theta_{2}}
+c_{22}^{(1)}K^{\theta_{2}}+c_{22}^{(1)}(\log K)\theta_{2}K^{\theta_{2}}.
\end{align*}

In general, for any $n\geq 1$,
\begin{equation}
\lambda m^{(n)}+nm^{(n-1)}=\mu\frac{\partial}{\partial x}m^{(n)}
+\frac{1}{2}\sigma^{2}\frac{\partial^{2}}{\partial x^{2}}m^{(n)}.
\end{equation}
It follows that
\begin{align*}
&m^{(n)}(\lambda,x)
\\
&=\begin{cases}
c_{11}^{(n)}e^{\theta_{1}x}+c_{12}^{(n)}xe^{\theta_{1}x}
+\cdots c_{1(n+1)}^{(n)}x^{n}e^{\theta_{1}x}
-\frac{(-1)^{n}n!}{\lambda^{n+1}}e^{x}+\frac{(-1)^{n}n!K}{\lambda^{n+1}} &\text{for $x\leq\log K$},
\\
c_{21}^{(n)}e^{\theta_{2}x}+c_{22}^{(n)}xe^{\theta_{2}x}
+\cdots c_{2(n+1)}^{(n)}x^{n}e^{\theta_{2}x} &\text{for $x>\log K$}.
\end{cases}
\end{align*}

For $i=1,2$, we have
\begin{align*}
&n\sum_{j=1}^{n}c_{ij}^{(n-1)}x^{j-1}e^{\theta_{i}x}
\\
&=-\lambda\sum_{j=2}^{n+1}c_{ij}^{(n)}x^{j-1}e^{\theta_{i}x}
+\mu\sum_{j=2}^{n+1}c_{ij}^{(n)}[x^{j-1}\theta_{i}e^{\theta_{i}x}+(j-1)x^{j-2}e^{\theta_{i}x}]
\\
&\qquad
+\frac{1}{2}\sigma^{2}\sum_{j=2}^{n+1}c_{ij}^{(n)}[x^{j-1}(\theta_{i})^{2}e^{\theta_{i}x}+(j-1)(j-2)x^{j-3}e^{\theta_{i}x}
+2(j-1)x^{j-2}\theta_{i}e^{\theta_{i}x}],
\end{align*}
with the understanding that $x^{-k}:=0$ for $k\leq -1$.

It follows that, for $i=1,2$,
\begin{align*}
&n\sum_{j=1}^{n}c_{ij}^{(n-1)}x^{j-1}
\\
&=-\lambda\sum_{j=2}^{n+1}c_{ij}^{(n)}x^{j-1}
+\mu\sum_{j=2}^{n+1}c_{ij}^{(n)}x^{j-1}\theta_{i}
+\mu\sum_{j=1}^{n}c_{i(j+1)}^{(n)}jx^{j-1}
\\
&\qquad
+\frac{1}{2}\sigma^{2}\sum_{j=2}^{n+1}c_{ij}^{(n)}x^{j-1}(\theta_{i})^{2}
+\frac{1}{2}\sigma^{2}\sum_{j=1}^{n-1}c_{i(j+2)}^{(n)}(j+1)jx^{j-1}
+\frac{1}{2}\sigma^{2}\sum_{j=1}^{n}c_{i(j+1)}^{(n)}2jx^{j-1}\theta_{i}
\\
&=-\lambda\sum_{j=2}^{n+1}c_{ij}^{(n)}x^{j-1}
+\mu\sum_{j=2}^{n+1}c_{ij}^{(n)}x^{j-1}\theta_{i}
+\mu\sum_{j=1}^{n}c_{i(j+1)}^{(n)}jx^{j-1}
\\
&\qquad
+\frac{1}{2}\sigma^{2}\sum_{j=2}^{n+1}c_{ij}^{(n)}x^{j-1}(\theta_{i})^{2}
+\frac{1}{2}\sigma^{2}\sum_{j=1}^{n-1}c_{i(j+2)}^{(n)}(j+1)jx^{j-1}
+\frac{1}{2}\sigma^{2}\sum_{j=1}^{n}c_{i(j+1)}^{(n)}2jx^{j-1}\theta_{i}
\\
&=-\lambda\sum_{j=2}^{n+1}c_{ij}^{(n)}x^{j-1}
+\mu\sum_{j=2}^{n+1}c_{ij}^{(n)}x^{j-1}\theta_{i}
+\mu\sum_{j=1}^{n}c_{i(j+1)}^{(n)}jx^{j-1}
\\
&\qquad
+\frac{1}{2}\sigma^{2}\sum_{j=2}^{n+1}c_{ij}^{(n)}x^{j-1}(\theta_{i})^{2}
+\frac{1}{2}\sigma^{2}\sum_{j=1}^{n-1}c_{i(j+2)}^{(n)}(j+1)jx^{j-1}
+\frac{1}{2}\sigma^{2}\sum_{j=1}^{n}c_{i(j+1)}^{(n)}2jx^{j-1}\theta_{i}.
\end{align*}
This implies that
\begin{equation}
nc_{i1}^{(n-1)}=\mu c_{i2}^{(n)}
+\sigma^{2}c_{i3}^{(n)}+\sigma^{2}c_{i2}^{(n)}\theta_{i},
\end{equation}
and for $j=2,3,\ldots,n$,
\begin{align*}
nc_{ij}^{(n-1)}
&=-\lambda c_{ij}^{(n)}
+\mu c_{ij}^{(n)}\theta_{i}
+\mu c_{i(j+1)}^{(n)}j
\\
&\qquad
+\frac{1}{2}\sigma^{2} c_{ij}^{(n)}(\theta_{i})^{2}
+\frac{1}{2}\sigma^{2} c_{i(j+2)}^{(n)}(j+1)j
+\frac{1}{2}\sigma^{2} c_{i(j+1)}^{(n)}2j\theta_{i},
\end{align*}
with $c_{i(n+2)}^{(n)}:=0$.

This gives us $c_{ij}^{(n)}$ for $j=2,3,\ldots,n+1$, $i=1,2$.
To determine $c_{11}^{(n)}$ and $c_{21}^{(n)}$,
we use the fact that $m^{(n)}(\lambda,x)$ is $C^{1}$ in $x$
and in particular, it is $C^{1}$ at $x=\log K$:
\begin{align*}
&\sum_{j=1}^{n+1}c_{1j}^{(n)}(\log K)^{j-1}K^{\theta_{1}}
=\sum_{j=1}^{n+1}c_{2j}^{(n)}(\log K)^{j-1}K^{\theta_{2}}\,,
\\
&\sum_{j=1}^{n+1}c_{1j}^{(n)}\left[(\log K)^{j-1}\theta_{1}K^{\theta_{1}}+(j-1)(\log K)^{j-2}K^{\theta_{1}}\right]
-\frac{(-1)^{n}n!}{\lambda^{n+1}}K
\\
&\qquad\qquad\qquad
=\sum_{j=1}^{n+1}c_{2j}^{(n)}\left[(\log K)^{j-1}\theta_{2}K^{\theta_{2}}+(j-1)(\log K)^{j-2}K^{\theta_{2}}\right]\,.
\end{align*}
This yields the equations \eqref{cn11} and \eqref{cn21}.
\end{proof}



\end{document}